\documentclass[11pt,english]{article}

\usepackage[margin=1in]{geometry}

\usepackage{varioref}
\usepackage{comment}
\usepackage[usenames,svgnames,xcdraw,table]{xcolor}
\definecolor{DarkBlue}{rgb}{0.1,0.1,0.5}
\definecolor{DarkGreen}{rgb}{0.1,0.5,0.1}

\usepackage{nicefrac}

\usepackage{hyperref}
\hypersetup{
   colorlinks   = true,
 linkcolor    = DarkBlue, 
 urlcolor     = DarkBlue, 
	 citecolor    = DarkGreen 
}

\usepackage{appendix}

\usepackage{mathtools}

\usepackage{nicefrac,bm}

\usepackage{algorithmic}
\usepackage{algorithm}
\usepackage{array}

\usepackage[font={small,it}]{caption}

\usepackage{graphicx,epsfig,amsmath,latexsym,amssymb,verbatim}

\usepackage[square,sort,semicolon,numbers]{natbib}

\newcommand{\extra}[1]{}

\usepackage{amsmath,amssymb,amsfonts}
\usepackage{amsthm}
\usepackage{thmtools,thm-restate}
\usepackage{enumitem}

\usepackage{xfrac}

\newtheorem{theorem}{Theorem}

\newtheorem{definition}{Definition}
\newtheorem{lemma}{Lemma}

\def\squareforqed{\hbox{\rlap{$\sqcap$}$\sqcup$}}
\def\qed{\ifmmode\squareforqed\else{\unskip\nobreak\hfil
\penalty50\hskip1em\null\nobreak\hfil\squareforqed
\parfillskip=0pt\finalhyphendemerits=0\endgraf}\fi}
\def\endenv{\ifmmode\;\else{\unskip\nobreak\hfil
\penalty50\hskip1em\null\nobreak\hfil\;
\parfillskip=0pt\finalhyphendemerits=0\endgraf}\fi}

\renewenvironment{proof}{\noindent \textbf{{Proof~} }}{\qed\medskip}
\newenvironment{proof+}[1]{\noindent \textbf{{Proof #1~} }}{\qed\medskip}
\newenvironment{remark}{\noindent \textit{{Remark.~}}}{\qed}

\mathchardef\ordinarycolon\mathcode`\:
\mathcode`\:=\string"8000
\def\vcentcolon{\mathrel{\mathop\ordinarycolon}}
\begingroup \catcode`\:=\active
  \lowercase{\endgroup
  \let :\vcentcolon
  }

\DeclareMathOperator*{\argmin}{arg\,min}
\DeclareMathOperator*{\argmax}{arg\,max}

\newcommand{\EF}{\sf{EF}}
\newcommand{\Eval}{{\sf Eval}}
\newcommand{\Cut}{{\sf Cut}}

\newcommand{\U}{\mathcal{U}}

\title{\bfseries Approximation Algorithms for  Envy-Free Cake Division with Connected Pieces}

\author{Siddharth Barman\thanks{Indian Institute of Science. {\tt barman@iisc.ac.in}} \and Pooja Kulkarni\thanks{University of Illinois at Urbana-Champaign. {\tt poojark2@illinois.edu}} }
\date{}

\begin{document}
\maketitle
\begin{abstract}
Cake cutting is a classic model for studying fair division of a heterogeneous, divisible resource among agents with individual preferences. Addressing cake division under a typical requirement that each agent must receive a connected piece of the cake, we develop approximation algorithms for finding envy-free (fair) cake divisions. In particular, this work improves the state-of-the-art additive approximation bound for this fundamental problem. Our results hold for general cake division instances in which the agents' valuations satisfy basic assumptions and are normalized (to have value $1$ for the cake). Furthermore, the developed algorithms execute in polynomial time under the standard Robertson-Webb query model. 

Prior work has shown that one can efficiently compute a cake division (with connected pieces) in which the additive envy of any agent is at most $1/3$. An efficient algorithm is also known for finding connected cake divisions that are (almost) $1/2$-multiplicatively envy-free. Improving the additive approximation guarantee and maintaining the multiplicative one, we develop a polynomial-time algorithm that computes a connected cake division that is both $\left(\frac{1}{4} +o(1) \right)$-additively envy-free and $\left(\frac{1}{2} - o(1) \right)$-multiplicatively envy-free. Our algorithm is based on the ideas of interval growing and envy-cycle-elimination. 

In addition, we study cake division instances in which the number of distinct valuations across the agents is parametrically bounded. We show that such cake division instances admit a fully polynomial-time approximation scheme for connected envy-free cake division. 
\end{abstract}

\section{Introduction}
Cake cutting is an exemplar of fair division literature \cite{brams1996fair,robertson1998cake,procaccia2015cake}. Since the foundational work of Steinhaus, Banach, and Knaster \cite{S48problem}, fair cake division has been extensively studied over decades, and it continues to inspire research, including algorithmic breakthroughs \cite{aziz2016discrete}, deep mathematical connections \cite{JPZ21,panina2021envy}, and applicable variants \cite{hosseini2020fair}. This fair-division model captures resource-allocation domains in which a divisible, heterogeneous resource (metaphorically, the cake) needs to be fairly divided among agents with individual, distinct preferences. For instance, cake division has been studied in the context of border negotiations \cite{brams1996fair} and fair electricity division \cite{baghel2022fair}. The predominant fairness notion of envy-freeness was also defined in the cake-division context \cite{foley1967resource}. This solution concept deems a cake division to be fair if each agent values the piece assigned to her over that of any other agent, i.e., if the agents are not envious of each other.

Formally, the cake is modeled as the interval $[0,1]$ and the cardinal preferences of the $n$ participating agents (over pieces of the cake) are expressed via valuation functions $v_1, \ldots, v_n$; in particular, $v_i(I) \in \mathbb{R}_+$ denotes the value that agent $i$ has for any interval $I \subseteq [0,1]$. In this work, we address cake division under the requirement that every agent must receive a connected piece (i.e., an interval) of the cake. That is, our goal is to partition the cake $[0,1]$ into exactly $n$ pairwise-disjoint intervals and assign them among the $n$ agents. This connectivity requirement is standard in literature and is motivated by practical settings wherein each agent must receive a contiguous part of the resource; consider, for instance, division of land or non-preemptive scheduling. Hence, in this setup, an envy-free (i.e., fair) division corresponds to a partition of $[0,1]$ into $n$ pairwise-disjoint intervals, $I_1, I_2, \ldots, I_n$, such that assigning each interval $I_i$ to agent $i \in [n]$ results in no envy, i.e., $v_i(I_i) \geq v_i(I_j)$, for all agents $i, j \in [n]$. 

The significance of envy-freeness is elevated by universal existential guarantees: under benign assumptions on agents' valuations, an envy-free cake division, in which each agent receives a connected piece, is guaranteed to exist \cite{stromquist1980cut,simmons1980private,edward1999rental}. Here, the elegant proof of Su~\cite{edward1999rental} is considered a foundational result across all of fair division. These strong existential results, however, do not have an algorithmic counterpart. Stromquist \cite{stromquist2008envy} has shown that even a finite-time algorithm does not exist for computing an envy-free cake division with connected pieces; this negative result holds in a model where the valuations are specified via an (adversarial) oracle. Furthermore, it is known that, under ordinal preferences, achieving envy-freeness with connected pieces is {\rm PPAD}-hard \cite{deng2012algorithmic}. 

These algorithmic barriers in route to finding exact envy-free cake divisions necessitate the study of approximation guarantees. The current paper contributes to this research thread by developing algorithms for finding connected cake divisions that are approximately envy-free. In particular, this work improves the state-of-the-art additive approximation bound for this fundamental fair division problem.

\paragraph{Our Results and Techniques.} Our algorithmic results hold for general cake division instances in which the agents' valuations satisfy basic assumptions and are normalized, such that the value for the entire cake for every agent is equal to one, i.e., $v_i([0,1])=1$ for all agents $i \in [n]$. Furthermore, the developed algorithms execute in the standard Robertson-Webb query model \cite{procaccia2015cake}.   

We address both additive and multiplicative approximations of envy-freeness. Specifically, for parameter $\varepsilon \in (0,1)$, a connected cake division $I_1,\ldots I_n$ (in which interval $I_i$ is assigned to agent $i \in [n]$) is said to be $\varepsilon$-envy-free ($\varepsilon$-$\EF$) iff no agent has more than $\varepsilon$ envy towards any other agent, i.e., $v_i(I_i) \geq v_i(I_j) - \varepsilon$ for all agents $i, j \in [n]$. Analogously, an $\alpha$-multiplicatively envy-free ($\alpha$-mult-$\EF$) cake division $I_1,\ldots I_n$ is one in which the envy is multiplicatively bounded within a factor of $\alpha$, i.e., $v_i(I_i) \geq \alpha v_i(I_j)$ for all agents $i, j \in [n]$; here parameter $\alpha \in (0,1]$.  

Our main result is a polynomial-time algorithm that computes a cake division (with connected pieces) that is simultaneously $\left(\frac{1}{4} + c \right)$-$\EF$ and $\left(\frac{1}{2} - c' \right)$-mult-$\EF$ (Theorems \ref{theorem:add-ef} and \ref{theorem:mult-ef}); here, $c$ and $c'$ are polynomially small (in $n$) terms. {For instance, our algorithm can be used to efficiently find a cake division that is $0.251$-$\EF$ and $0.499$-mult-$\EF$.}  

Our result improves upon the previously best known additive approximation guarantee. Specifically, prior work of Goldberg et al.~\cite{goldberg2020contiguous} provides an efficient algorithm for computing a $\frac{1}{3}$-$\EF$ cake division (with connected pieces); here, the computed allocation can leave some agents with no cake allocated to them and, hence, incur unbounded multiplicative envy. On the multiplicative front, for a lower order term $\kappa$, Arunachaleswaran et al.~\cite{arunachaleswaran2019fair} obtain a $\left(\frac{1}{2} - \kappa \right)$-mult-$\EF$ guarantee, in conjunction with an additive envy bound close to $\frac{1}{3}$. Therefore, for envy-free cake division, we improve the additive approximation guarantee from $\frac{1}{3}$ to (almost) $\frac{1}{4}$, while maintaining the best known multiplicative one.  

Our algorithm extends the interval-growing method of \cite{arunachaleswaran2019fair} with the idea of bifurcating intervals (see Definition \ref{defn:bifurcating}). Such intervals satisfy the property that if an agent $i$ receives an interval that is bifurcating with respect to $v_i$, then irrespective of how the rest of the cake is assigned, agent $i$'s envy towards others remains bounded. For the algorithm's design and analysis, we modify each agent's valuation to have a preference for bifurcating intervals. With these modified valuations, we build upon the idea of interval growing. In particular, we first obtain an allocation of (pairwise disjoint) intervals that might partially cover the entire cake, though induce bounded envy among the agents and against the unassigned intervals. Then, we use an envy-cycle-elimination idea to further allocate small pieces till at most $n$ unassigned intervals remain. Envy-graphs and the cycle-elimination method have been extensively utilized in fair division; see, e.g., \cite{lipton2004approximately}. However, their use for \emph{contiguous} cake cutting (i.e., division under the contiguity requirement) is novel. We employ cycle elimination in such a way that envy remains bounded as we allocate more and more of the cake. Finally, we have $n$ assigned and at most $n$ unassigned intervals. We pair up adjacent  assigned and unassigned intervals to overall obtain a complete partition of the cake that has bounded envy; see Section \ref{subsec:algorithm-description} for a  detailed description of the algorithm. 

It is relevant to note the technical distinctions between the algorithm of Arunachaleswaran et al.~\cite{arunachaleswaran2019fair} and the current one. In contrast to the prior work, the current algorithm executes with a novel modification of the valuations (to incorporate preferences towards bifurcating intervals). Moreover, the current analysis is more involved; in particular, the analysis requires multiple new lemmas and consideration of intricate cases (see, e.g., Lemmas \ref{lemma:hat-monotone} to \ref{lemma:bifurcating-phase-two} and the case analysis in the proof of Theorem \ref{theorem:add-ef}).

Our second result addresses cake division instances in which the number of distinct valuations is bounded. Specifically, we consider instances in which, for a parameter $\varepsilon \in (0,1)$ and across the $n$ agents, the number of \emph{distinct} valuations is at most $(\varepsilon n-1)$. For such instances with bounded heterogeneity, we provide an algorithm that computes $\varepsilon$-$\EF$ allocations in time polynomial in $n$ and $\frac{1}{\varepsilon}$ (Theorem \ref{theorem:bounded-het}). Note that such settings naturally generalize the case of identical valuations. Fair division algorithms under identical valuations have been developed in many contexts (beyond cake division). Our result shows that, under this natural generalization, a strong additive approximation guarantee can be obtained for connected envy-free cake division; see Section \ref{sec:boundedhet}.\footnote{We also detail at the end of the Section \ref{sec:boundedhet} that achieving multiplicative approximation bounds for envy under bounded heterogeneity is as hard as it is in the general case.} 

\paragraph{Additional Related Work.} Prior works in (connected) cake division have also studied improved approximation guarantees for specific valuation classes. For instance, it is shown in \cite{barman2021fair} that a connected cake division with arbitrarily small envy can be computed efficiently if the agents' value densities satisfy the monotone likelihood ratios property. Another studied valuation class is that of single-block valuations; in particular,  these correspond to valuations in which the agents have a constant density over some (agent-specific) interval of cake and zero everywhere else. The work of Alijani et al.~\cite{alijani2017envy} provides an efficient algorithm for finding (exact) envy-free cake division under single-block valuations that satisfy an ordering property. For arbitrary single-block valuations (without the ordering property), Goldberg et al.~\cite{goldberg2020contiguous} obtain a $\frac{1}{4}$-$\EF$ guarantee. They also obtain {\rm NP}-hardness results for connected envy-free cake division under additional constraints, such as conforming to a given cut point.  

Focussing on query complexity, Br{\^a}nzei and Nisan \cite{branzei2017query} show that an $\varepsilon$-$\EF$ cake division (with connected pieces) can be computed in a query efficient manner but the algorithm runs in time exponential in $n$ and $\frac{1}{\varepsilon}$. By contrast, we develop polynomial-time algorithms.

The study of approximation guarantees---to bypass computational or existential barriers---is an established research paradigm in theoretical computer science. For instance, in discrete fair division, (multiplicative) approximation bounds for the maximin share has received significant attention in recent years; see \cite{amanatidis2022fair} and multiple references therein. Also, in algorithmic game theory, approximation guarantees for Nash equilibria in two-player games have been extensively studied; see, e.g., \cite{daskalakis2007progress,tsaknakis2007optimization,kontogiannis2006polynomial}. Our work contributes to this thematic thread with a focus on cake division. \\

\noindent
\emph{Non-Contiguous Cake Division.} Algorithmic aspects of (exact) envy-free cake division remain challenging even without the connectivity requirement. In fact, the existence of a finite-time algorithm for noncontiguous envy-free cake division remained open until the work of Brams and Taylor \cite{brams1995envy}. For noncontiguous envy-free cake divisions, an explicit runtime bound---albeit a hyper-exponential one---was obtained in the notable work of Aziz and Mackenzie \cite{aziz2016discrete}. Prior works have also addressed non-contiguous envy-free cake division for special valuation classes: \cite{alijani2017envy} obtains a polynomial-time algorithm for finding (exact, but not necessarily contiguous) envy-free cake divisions under single-block valuations. Also, \cite{wang2019cake} develops a polynomial-time algorithm for computing non-contiguous envy-free cake divisions under single-peaked preferences. For the non-contiguous setting, \cite{lipton2004approximately} provides a fully polynomial-time approximation scheme for computing approximate envy-free cake divisions.
\section{Notation and Preliminaries}
\label{section:notation}

We consider fair division of a divisible, heterogeneous good---i.e., a cake---among $n$ agents. The cake is modeled as the interval $[0,1]$, and the cardinal preferences of the agents $i \in [n]$ over the cake are expressed via valuation functions $v_i$. In particular, $v_i(I) \in \mathbb{R}_+$ denotes the valuation that agent $i \in [n]$ has for any interval $I = [x, y] \subseteq [0,1]$; here, $0 \leq x \leq y \leq 1$. As in prior works (see, e.g.,~\cite{procaccia2015cake}), we will address valuations $\{v_i\}_{i=1}^n$ that are (i) nonnegative: $v_i(I) \geq 0$ for all intervals $I \subseteq [0,1]$, (ii) normalized: $v_i([0,1]) = 1$ for all agents $i$, (iii) divisible: for any interval $[x,y] \subseteq [0,1]$ and scalar $\lambda \in [0,1]$, there exists a point $z \in [x,y]$ with the property that $v_i([x,z]) = \lambda v_i([x,y])$, and (iv) additive: $v_i(I \cup J) = v_i(I) + v_i(J)$ for any pair of disjoint intervals $I, J \subseteq [0,1]$. 

These properties ensure that for all agents $i \in [n]$ and each interval $I\subseteq [0,1]$ we have $0 \leq v_i(I) \leq 1$. Also, note that, since the valuations $v_i$ are divisible, they are non-atomic: $v_i([x,x])=0$ for all points $x \in [0,1]$. Relying on this property, we will throughout regard, as a convention, two intervals to be disjoint even if they intersect exactly at an endpoint. Our algorithms  efficiently execute in the standard Robertson-Webb query model \cite{robertson1998cake}, that provides access to the agents' valuations via the following queries: \\

\noindent
(i) Evaluation queries, ${\Eval}$$_i(x,y)$: Given points $0 \leq x \leq y \leq 1$, the oracle returns the value that agent $i$ has for the interval $[x,y]$, i.e., returns $v_i([x,y])$. \\
\noindent 
(ii) Cut queries, ${\Cut}$$_i(x, \nu)$: Given an initial point $x \in [0,1]$ and a value $\nu \in (0,1)$, the oracle returns the leftmost point $y \in [x,1]$ with the property that $v_i([x,y]) \geq \nu$. If no such $y$ exists, the response to the query is $1$.

\paragraph{Allocations.} The current work address fair cake division under the requirement that each agent must receive a connected piece. That is, we focus solely on assigning to each agent a single sub-interval of $[0,1]$. Specifically, in a cake division instance with $n$ agents, an allocation is defined as an $n$-tuple of pairwise-disjoint intervals, $\mathcal{I} = (I_1, I_2, \ldots, I_n)$, where interval $I_i$ is assigned to agent $i \in [n]$ and $\bigcup_{i=1}^n I_i = [0,1]$.  In addition, we will use the term partial allocation to refer to an $n$-tuple of pairwise-disjoint intervals $\mathcal{P} = (P_1, \ldots, P_n)$ that do not necessarily cover the entire cake, $\cup_{i \in [n]} P_i \subsetneq [0,1]$. Note that, in an allocation $\mathcal{J} = (J_1, J_2, \ldots, J_n)$, partial or complete, each interval $J_i$ is indexed to identify the agent $i$ that owns the interval, and not how the intervals are ordered within $[0,1]$. 

Furthermore, for a partial allocation $\mathcal{P}=(P_1, \ldots, P_n)$, write $\U_\mathcal{P}= \left\{U_1, \ldots, U_t \right\}$ to denote the collection of unassigned intervals that remain after the assigned ones (i.e., $P_i$s) are removed from $[0,1]$. Formally, $\U_\mathcal{P} = \left\{U_1, \ldots, U_t \right\}$ is the minimum-cardinality collection of disjoint intervals that satisfy $\bigcup_{i} U_i = \left[0,1\right] \setminus \left( \bigcup_{j=1}^n P_j \right)$. 
 
 \paragraph{Approximate Envy-Freeness.} The fairness notions studied in this work are defined next. An allocation $\mathcal{E} = (E_1, \ldots, E_n)$ is said to be envy free ($\EF$) iff each agent prefers the interval assigned to her over that of any other agent, $v_i(E_i) \geq v_i(E_j)$ for all agents $i, j \in [n]$. This paper addresses both additive and multiplicative approximations of envy freeness. 

\begin{definition}[$\varepsilon$-$\EF$]
In a cake division instance with $n$ agents and for a parameter $\varepsilon \in (0,1)$, an (partial) allocation $\mathcal{I} = (I_1, I_2, \ldots, I_n)$ is said to be $\varepsilon$-additively envy-free ($\varepsilon$-$\EF$) iff $v_i(I_i) \geq v_i(I_j) - \varepsilon$, for all agents $i, j \in [n]$. 
\end{definition}

\begin{definition}[$\alpha$-mult-$\EF$]
For a parameter $\alpha \in (0,1)$, an (partial) allocation $\mathcal{I} = (I_1, I_2, \ldots, I_n)$  is said to be $\alpha$-multiplicatively envy-free ($\alpha$-mult-$\EF$) iff, for all agents $i,j \in [n]$, we have $v_i(I_i) \geq  \alpha \ v_i(I_j)$.
\end{definition}
\section{Approximation Algorithm for Envy-Free Cake Division}
\label{sec:mainsec}

This section develops an algorithm for efficiently computing a cake division (with connected pieces) that is $\left( \frac{1}{4} + o(1) \right)$-$\EF$ and $\left(\frac{1}{2} - o(1) \right)$-mult-$\EF$. 

For the design of the algorithm, we will use the notion of bifurcating intervals. For an agent $i$, a bifurcating interval $X$ satisfies the property that, if $i$ is assigned interval $X$, then one can divide the rest of the cake in any way and still agent $i$ will have at most $1/4$ envy towards any other agent. Formally,
\begin{definition}[Bifurcating Intervals] \label{defn:bifurcating}
An interval $[x, y] \subseteq [0,1]$ is said to be a bifurcating interval for an agent $i\in [n]$ iff
\begin{align*}
v_i([x,y]) \geq \frac{1}{4}, \qquad v_i([0, x]) \leq \frac{1}{2}, \qquad \text{and} \quad v_i([y,1]) \leq \frac{1}{2}.
\end{align*}
\end{definition}
For each agent $i \in [n]$, we extend the valuation $v_i$ to a function $\widehat{v}_i$ which codifies a preference towards bifurcating intervals. Formally, for each agent $i \in [n]$ and any interval $X \subseteq [0,1]$, define
\begin{align}
\widehat{v}_i(X) \coloneqq \begin{cases}
		1 \ & \text{  if } X \text{ is bifurcating for } i. \\
		 v_i(X) \ & \text{ if } X \text{ is not bifurcating for } i.
	 \end{cases} \label{defn:hat-v}
\end{align}
We note that, in contrast to the valuation $v_i$, the function $\widehat{v}_i$ is not additive.\footnote{Also, the function $\widehat{v}_i$ is not divisible.} However, analogous to $v_i$, the function $\widehat{v}_i$ is monotonic, normalized, and nonnegative. In addition, given access to $\Eval_i()$ queries, we can efficiently compute $\widehat{v}_i(X)$ for any interval $X \subseteq [0,1]$. We will show in the analysis that the algorithm's steps involving $\widehat{v}_i$s can be implemented efficiently, given Robertson-Webb query access to the underlying valuations. The claim below provides a bound on the value of non-bifurcating intervals. 
\begin{restatable}{claim}{ClaimNB}
\label{claim:non-bifurcating}
For any agent $i \in [n]$, if $Y \subseteq [0,1]$ is \emph{not} a bifurcating interval, then, $\widehat{v}_i(Y) = v_i(Y) < \frac{1}{2}$.
\end{restatable} 
\begin{proof}
All intervals $H \subseteq [0,1]$ of value $v_i(H) \geq \frac{1}{2}$ are bifurcating; see Definition \ref{defn:bifurcating} and recall that the agents' valuations are normalized, $v_i([0,1]) = 1$. Hence, for any non-bifurcating interval $Y \subseteq [0,1]$ we have $v_i(Y) < \frac{1}{2}$, i.e., $\widehat{v}_i(Y) = v_i(Y) < \frac{1}{2}$ (see equation (\ref{defn:hat-v})). 
\end{proof}

We will also use the construct of an envy-graph. Specifically, for a partial allocation $\mathcal{P}=(P_1, \ldots, P_n)$, an envy-graph $G_\mathcal{P}$ is a directed graph over $n$ vertices. Here, the vertices represent the $n$ agents and a directed edge, from vertex $i$ to vertex $j$, is included in the graph iff 
$\widehat{v}_i (P_i) < \widehat{v}_i(P_j)$. Envy-graphs and the cycle-elimination algorithm (detailed next) has been extensively utilized in discrete fair division; see, e.g., \cite{lipton2004approximately}. However, their use for cake cutting is novel.   

We will next show that, if for any partial allocation $\mathcal{P}=(P_1, \ldots, P_n)$, the envy-graph $G_\mathcal{P}$ contains a cycle, then we can in fact resolve the cycle---by reassigning the intervals---and eventually obtain a partial allocation $\mathcal{Q} = (Q_1, \ldots, Q_n)$ whose envy-graph $G_\mathcal{Q}$ is acyclic.\footnote{Here, the reassignment of the intervals implies that there exists a permutation $\pi \in \mathbb{S}_n$ such that $Q_i = P_{\pi(i)}$ for all agents $i$.}
\begin{restatable}{lemma}{lemCycleElimination} \label{lemma:cycle-elimination}
Given any partial allocation $\mathcal{P}= (P_1, \ldots, P_n)$, one can reassign the intervals $P_i$s among the agents and efficiently find another partial allocation $\mathcal{Q} = (Q_1, \ldots, Q_n)$ with the properties that
\begin{itemize}
\item[(i)] The envy-graph $G_\mathcal{Q}$ is acyclic.
\item[(ii)] The value $\widehat{v_i}(Q_i) \geq \widehat{v}_i(P_i)$, for all agents $i \in [n]$.
\end{itemize}
\end{restatable}

The proof of Lemma \ref{lemma:cycle-elimination} is standard and, for completeness, is provided in Appendix \ref{app:mainsec}. \\

Recall that, for any partial allocation $\mathcal{P}=(P_1, \ldots, P_n)$, the set $\U_\mathcal{P}  
= \left\{U_1, \ldots, U_t \right\}$ denotes the collection of unassigned intervals that remain after the intervals $P_i$s are removed from $[0,1]$. Also, note that for any partial allocation $\mathcal{P}= (P_1, \ldots, P_n)$, we have $|\mathcal{U}_\mathcal{P}| \leq n+1$.

\subsection{Interval Growing and Cycle Elimination}\label{subsec:algorithm-description}
\begin{algorithm}[h]
\caption{Approximation Algorithm for Connected Cake Division} \label{alg:quat-ef}
\textbf{Input:} A cake division instance with oracle access to the valuations $\{v_i\}_{i=1}^n$ of the $n$ agents and a fixed constant $\delta \in (0,1)$.  \\ 
\textbf{Output:} A complete allocation $(I_1,\ldots, I_n)$. 
\begin{algorithmic}[1]
				\STATE Initialize partial allocation $\mathcal{P}=(P_1, \ldots, P_n)=(\emptyset, \ldots, \emptyset)$ and $\mathcal{U}_\mathcal{P} = \left\{ [0,1] \right\}$.  
				\WHILE{there exists  an unassigned interval ${U} =[\ell, r] \in \mathcal{U}_{\mathcal{P}}$ and an agent $i \in [n]$ such that $\widehat{v}_i (U) \geq \widehat{v}_i(P_i) + \frac{\delta}{n}$} \label{step:while-loop}
				\STATE \label{step:candidates} Let $C \coloneqq \left\{ i \in [n] \ : \widehat{v}_i ({U}) \geq \widehat{v}_i(P_i) + \frac{\delta}{n} \right\}$ and, for every agent $i \in C$, set $r_i \in [{\ell}, {r}] $ to be the leftmost point such that $\widehat{v}_i ([\ell, r_i]) \geq \widehat{v}_i(P_i) + \frac{\delta}{n}$.
				\STATE \label{step:update} Select agent $a \in \argmin_{i \in C} \  r_i$ and update the partial allocation $\mathcal{P}$: assign $P_{a}  \leftarrow [\ell, r_{a}]$ and keep the interval assignment of all other agents unchanged.
				\STATE \label{step:updateU} Update $\mathcal{U}_{\mathcal{P}}$ to be the collection of intervals that are left unassigned under the current partial allocation $\mathcal{P}$.
				\ENDWHILE \label{step:while-ends}
		\WHILE{$|\mathcal{U}_\mathcal{P}| > n$} \label{step:begin-phase-two}
		\STATE \label{step:envycycle} Update $\mathcal{P}=(P_1,\ldots, P_n)$ following Lemma \ref{lemma:cycle-elimination} to ensure that the envy-graph $G_\mathcal{P}$ is acyclic. 
		\STATE \label{step:source-selection} Let $s \in [n]$ be a source vertex in the graph $G_\mathcal{P}$, with assigned interval $P_s = [\ell_s, r_s]$.
		\STATE \label{step:selectinterval} Let $\widetilde{U} =[r_s, \widetilde{r}]  \in \mathcal{U}_\mathcal{P}$ be the unassigned interval that is adjacent (on the right) to $P_s$. 
		\COMMENT{Since $|\mathcal{U}_\mathcal{P}| > n$, such an interval $\widetilde{U}$ is guaranteed to exist.}
		\STATE \label{step:crumb} Write $x \in [r_s, \widetilde{r}]$ to be the point with the property that $v_i([r_s, x]) \leq \frac{\delta}{n}$ for all agents $i$ and this inequality is tight for at least one agent. Append $P_s \leftarrow P_s \cup [r_s, x]$.	\\
		\COMMENT{If for all agents the value of $\widetilde{U}$ is at most $\frac{\delta}{n}$, then append $P_s \leftarrow P_s \cup \widetilde{U}$.}
		\ENDWHILE \label{step:end-phase-two}
		\STATE Index the unassigned intervals $U_j \in \mathcal{U}_\mathcal{P}$ such that each $U_j$ is adjacent to a \emph{distinct} interval $P_j$, for all $j$. \COMMENT{Since $|\mathcal{U}_\mathcal{P}| \leq n$, such an indexing is possible.}
		\STATE For all agents $i$, set interval $I_i = P_i \cup U_i$. \COMMENT{If an unassigned interval is not associated with $P_i$, then set $I_i = P_i$.} 
		\RETURN allocation $\mathcal{I} = (I_1, \ldots, I_n)$.				
\end{algorithmic}
\end{algorithm}

Our algorithm (Algorithm \ref{alg:quat-ef}) consists of two phases. In Phase ${\rm I}$ (Lines \ref{step:while-loop} to \ref{step:while-ends} in Algorithm \ref{alg:quat-ef}), which we call the interval growing phase, the algorithm starts with empty intervals---i.e., $P_i = \emptyset$ for all $i$---and iteratively grows these intervals while maintaining bounded envy among the agents. In particular, to extend a partial allocation $\mathcal{P}=(P_1,\ldots,P_n)$, we first judiciously select an unassigned interval ${U} \in \mathcal{U}_\mathcal{P}$ and then assign an inclusion-wise minimal sub-interval of $U$ to an agent $a$. The sub-interval of $U$ and agent $a$ are selected such that the function value, $\widehat{v}_a$, increases appropriately and, at the same time, the envy towards $a$ (from any other agents) remains bounded. Note that, in this phase, the cake might not be allocated completely, but the invariant of bounded envy is maintained throughout. Phase {\rm I} terminates with a partial allocation $\overline{\mathcal{P}}=(\overline{P}_1,\ldots, \overline{P}_n)$ under which each agent $i \in [n]$ has bounded envy towards the other agents and towards all the unassigned intervals $U \in \mathcal{U}_{\overline{\mathcal{P}}}$ (see Lemma \ref{lem:envyfreephase1}). 

At the end of Phase {\rm I}, it is possible that the number of unassigned intervals is $n+1$. The objective of Phase {\rm II} (Lines \ref{step:begin-phase-two} to \ref{step:end-phase-two} in the algorithm) is to reduce the number of unassigned intervals, while maintaining bounded envy between the agents and against the unassigned intervals. Towards this, we use the cycle-elimination method (Lemma \ref{lemma:cycle-elimination}) to first ensure that for the maintained partial allocation $\mathcal{P}$ the envy-graph $G_\mathcal{P}$ is acyclic. Now, given that the directed graph $G_\mathcal{P}$ is acyclic, it necessarily admits a source vertex $s \in [n]$, i.e., a vertex $s$ with no incoming edges. Furthermore, by the definition of the envy graph, we get that no agent has sufficiently high envy towards the source vertex $s \in [n]$. With this guarantee in hand, we enlarge the interval assigned to $s$ (i.e., enlarge $P_s$) while maintaining bounded envy overall. Specifically, we append to $P_s$ a piece of small enough value from the unassigned interval $\widetilde{U}$ adjacent to $P_s$. Since $|\mathcal{U}_\mathcal{P}|=n+1$, an unassigned interval, adjacent to $P_s$, is guaranteed to exist. Also, note that this extension ensures that $P_s$ continues to be a connected piece of the cake, i.e., agent $s$ continues to receive a single interval. Performing such updates, Phase {\rm II} efficiently finds a partial allocation $\mathcal{P}$ with the property that $|\mathcal{U}_\mathcal{P}| \leq n$. Since at the end of Phase {\rm II} the number of unassigned intervals is at most $n$, we can associate each unassigned interval $U \in \mathcal{U}_\mathcal{P}$ with a \emph{distinct} assigned interval $P_j$ that is adjacent to $U$. We merge each assigned interval $P_i$ with the associated and adjacent unassigned interval $U_i$ (if any) to obtain the interval $I_i$ for each agent $i \in [n]$. 
The intervals $I_1, I_2, \ldots, I_n$ completely partition the cake $[0,1]$ and constitute the returned allocation $\mathcal{I} = (I_1, I_2, \ldots, I_n)$. We will establish in Section \ref{section:analysis} that $\mathcal{I}$ satisfies the stated approximation guarantees for envy-freeness.  In Section \ref{section:runtime}, we will prove that the two phases run in polynomial time (Lemma \ref{lemma:time-complexity}) under the Robertson-Webb query model. 

\subsection{Runtime Analysis}
\label{section:runtime}
We begin by noting that, given Robertson-Webb query access to the underlying valuation $v_i$s, we can answer cut and evaluation queries for the functions $\widehat{v}_i$ (see equation (\ref{defn:hat-v})) in polynomial time. That is, given points $0 \leq x \leq y \leq 1$, we can find $\widehat{v}_i([x,y])$ in polynomial time. Also, given a point $x\in [0,1]$ and value $\nu \in [0,1]$, we can efficiently compute the leftmost point $y$ (if one exists) that satisfies $\widehat{v}_i([x,y]) \geq \nu$. The proof of this claim is deferred to Appendix \ref{app:mainsec}. 

\begin{restatable}{claim}{ClaimRWHatV}
\label{clm:rwhatvs}
For any agent $i \in [n]$, given Robertson-Webb query access to the valuation $v_i$, we can answer cut and evaluation queries with respect to the function $\widehat{v}_i$ in polynomial time.
\end{restatable}

We now prove that the algorithm executes in polynomial time.
\begin{lemma}
\label{lemma:time-complexity}
Given a fixed constant $\delta \in \left(0,\frac{1}{4}\right)$ and any cake division instance with (Robertson-Webb) query access to the valuations of the $n$ agents, Algorithm \ref{alg:quat-ef} computes an allocation in time that is polynomial in $n$ and $\frac{1}{\delta}$. 
\end{lemma}
\begin{proof}
We will first establish the time complexity of Phase {\rm I} of the algorithm. Note that, in every iteration of this phase (i.e., in every iteration of the while-loop between Lines \ref{step:while-loop} and \ref{step:while-ends}), for some agent $a \in [n]$, the value $\widehat{v}_a(P_a)$ increases additively by at least $\frac{\delta}{n}$; see Lines \ref{step:candidates} and \ref{step:update}. Since the functions $\widehat{v}_i$s are monotonic and upper bounded by $1$, the first while-loop in the algorithm iterates at most $\frac{n^2}{\delta}$ times. We next show that each iteration of this while-loop can be implemented in polynomial time and, hence, obtain that overall Phase {\rm I} executes in polynomial time.  Note that the execution condition of the while-loop (Line \ref{step:while-loop}) can be evaluated efficiently, since the evaluation query under $\widehat{v}_i$s can be answered in polynomial time (Claim \ref{clm:rwhatvs}). Similarly, the candidate set $C$ in Line \ref{step:candidates} can be computed efficiently. Finding the points $r_i$s in Line \ref{step:candidates} entails answering cut queries for the functions $\widehat{v}_i$s and this too can be implemented efficiently (Claim \ref{clm:rwhatvs}). Therefore, all the steps in the while-loop can be implemented efficiently, and we get that Phase {\rm I} terminates in polynomial time. 

For Phase {\rm II} and each maintained partial allocation $\mathcal{P}=(P_1, \ldots, P_n)$, consider the potential $\varphi(\mathcal{P}) \coloneqq \sum_{i=1}^n \sum_{j=1}^n v_i(P_j)$. In each iteration of the second while-loop (Lines \ref{step:begin-phase-two} to \ref{step:end-phase-two}), the assigned region of the cake (i.e., $\cup_{i\in [n]} P_i$) monotonically increases. Indeed, while updating a partial allocation, the intervals might get reassigned among the agents, however, the union $\cup_{i\in [n]} P_i$ increases in each iteration of the second while-loop. Furthermore, in every iteration, for at least one agent $i$ and the selected interval $P_s$ (see Line \ref{step:crumb}), the value increases by $\frac{\delta}{n}$.\footnote{If the unassigned interval $\widetilde{U}$, considered in Line \ref{step:crumb}, is of value less than $\frac{\delta}{n}$ for all agents $i \in [n]$, then after that update the number of unassigned intervals (i.e., $|\mathcal{U}_\mathcal{P}|$) strictly  decreases. Hence, after such an update, the while-loop terminates.} Hence, in each iteration, the potential $\varphi$ 
increases by at least $\frac{\delta}{n}$. Also, note that the potential is upper bounded by $n$ and, hence, the second while-loop iterates at most $\frac{n^2}{\delta}$ times. Since all the steps in each iteration of the loop can be implemented in polynomial time---including the envy cycle elimination one (Lemma \ref{lemma:cycle-elimination})---we get that Phase {\rm II} itself executes in polynomial time. 

The final merging of the intervals takes linear time. This, overall, establishes the polynomial-time complexity of the algorithm.
\end{proof}

\subsection{Approximation Guarantee}
\label{section:analysis}

We first note a monotonicity property with respect to the function values, $\widehat{v}_i$s, satisfied during the execution of the algorithm. 
\begin{restatable}{lemma}{LemmaHatVMonotone}
\label{lemma:hat-monotone}
For any agent $i \in [n]$, the function values $\widehat{v}_i$ of the assigned intervals, $P_i$s, are nondecreasing through the execution of Algorithm \ref{alg:quat-ef}. 
\end{restatable} 
\begin{proof}
To establish the monotonicity under $\widehat{v}_i$ in Phase {\rm I}, consider any iteration for the first while-loop. Here, for the selected agent $a \in [n]$, the value of the assigned interval, under $\widehat{v}_a$, in fact increases and for all the other agents it continues to be the same. Hence, the lemma holds throughout Phase {\rm I}. The monotonicity is also maintained during the execution of Phase {\rm II}: the value under $\widehat{v}_i$ does not decrease in Line \ref{step:envycycle} (Lemma \ref{lemma:cycle-elimination}) or in Line \ref{step:crumb}. Therefore, the lemma stands proved. 
\end{proof}

Next, we  assert that, throughout the execution of the algorithm, the assigned intervals satisfy an inclusion-wise minimality property. Note that in the following lemma we evaluate agent $i$'s assigned interval under the function $\widehat{v}_i$ and evaluate the compared interval $X$ under the valuation $v_i$. 

\begin{restatable}{lemma}{LemmaMinInclusion}
\label{lemma:min-inclusion}
Let $\mathcal{P}'=(P'_1, \ldots, P'_n)$ be any partial allocation considered during the execution of Algorithm \ref{alg:quat-ef}. Then, for any two assigned intervals $P'_i$ and $P'_j=[\ell'_j, r'_j]$ along with any strict subset $X = [\ell'_j, x] \subsetneq P'_j$ (i.e., $x < r'_j$), we have  ${v}_i(X) < \widehat{v}_i(P'_i) + \frac{\delta}{n}$.
\end{restatable}
\begin{proof}
We establish the lemma via an inductive argument. Indeed, the initial partial allocation $(\emptyset, \ldots, \emptyset)$ satisfies the desired property. Now, consider any iteration of the first-while loop, and write $\mathcal{P}''=(P''_1, \ldots, P''_n)$ to be the partial allocation that gets updated (in this iteration) to $\mathcal{P}'=(P'_1, \ldots, P'_n)$. In particular, let $a$ be the agent selected in Line \ref{step:update}. Note that for all the other agents $i \neq a$, the assigned interval remains unchanged, $P'_i = P''_i$. Also, the induction hypothesis implies that $\mathcal{P}''$ satisfies the lemma. Hence, for all the agents $i, j \neq a$ (whose assigned intervals have not changed), the desired property continues to hold. Furthermore, agent $a$ receives an interval of higher function value, $\widehat{v}_a(P'_a) \geq \widehat{v}_a(P''_a) + {\delta}/{n}$. Hence, we have the lemma from agent $a$ against any other agent $j$. 

It remains to show that the lemma holds between $P'_i$ and $P'_a = [\ell'_a, r'_a]$. Assume, towards a contradiction, that there exists a strict subset $X = [{\ell}'_a, x] \subsetneq {P}'_a$ such that ${v}_i(X) \geq \widehat{v}_i(P'_i) + {\delta}/{n}$. Since $P'_i = P''_i$ and $\widehat{v}_i(X) \geq v_i(X)$, we obtain $\widehat{v}_i(X)  \geq \widehat{v}_i(P''_i) + {\delta}/{n}$. This, however, contradicts the selection criterion in Lines \ref{step:candidates} and \ref{step:update}. In particular, this bound implies $r_i < r_a$ (see Line \ref{step:candidates}) and, hence, $a$ would not be the selected agent in Line \ref{step:update}. Therefore, by way of contradiction, we have that the property holds with respect to $P'_a$ as well. 

The above-mentioned arguments prove that the lemma holds for all allocation considered in Phase {\rm I}. Next, we show that it continues to hold through Phase {\rm II}. 

Consider any iteration of the second while-loop, and write $\mathcal{P}''=(P''_1, \ldots, P''_n)$ to be the partial allocation that gets updated in this iteration. The induction hypothesis gives us that $\mathcal{P}''$ satisfies the desired property. In Line \ref{step:envycycle} the intervals are reassigned among the agents (i.e., the collection of intervals remains unchanged) and for each agent $i$, the value, under $\widehat{v}_i$, of the assigned interval does not decrease; see Lemma \ref{lemma:cycle-elimination}. Hence, the property continues to hold after Line \ref{step:envycycle}. For analyzing the rest of the iteration, let $s \in [n]$ denote the (source) agent that gets selected in Line \ref{step:source-selection} and $P'_s$ be the updated interval for agent $s$; in particular, interval $P'_s$ is obtained by appending a piece to $P''_s$. Since $s$ is the only agent whose interval got updated here, the lemma continues to hold between all other agents $i, j \neq s$. Also, the property is maintained from agent $s$'s perspective, since $\widehat{v}_s(P'_s) \geq \widehat{v}_s(P''_s)$. To complete the proof we will next show that the property is upheld between $P'_i$ and $P'_s$, for any $i \in [n]$. 

Note that, for agent $i \neq s$, the assigned interval remains unchanged during the current update, $P'_i = P''_i$. Furthermore, the fact that $s$ is a source vertex gives us
\begin{align}
\widehat{v}_i(P''_i)  \geq \widehat{v}_i(P''_s)  \geq v_i(P''_s)  \label{ineq:prime}
\end{align}
The extension of $P''_s$ to $P'_s$ (performed in Line \ref{step:crumb}) ensures that $v_i(P'_s) \leq v_i(P''_s) + \delta/n$. Hence, inequality (\ref{ineq:prime}) gives us $v_i(P'_s) \leq \widehat{v}_i(P'_i) + \delta/n$. That is, there does not exist an $X \subsetneq P'_s$ with the property that ${v}_i(X) \geq \widehat{v}_i(P'_i) + {\delta}/{n}$. This completes the proof. 
\end{proof}

The next lemma provides a bounded envy guarantee for the partial allocation $\overline{\mathcal{P}}=(\overline{P}_1,\ldots, \overline{P}_n)$ computed by Phase {\rm I}. Note that in this lemma, while considering envy from agent $i$ to agent $j$, we evaluate $\overline{P}_i$ with respect to $\widehat{v}_i$ and evaluate $\overline{P}_j$ under $v_i$.  

\begin{lemma}\label{lem:envyfreephase1}
Let $\overline{\mathcal{P}}=(\overline{P}_1, \ldots, \overline{P}_n)$ be the partial allocation maintained by Algorithm \ref{alg:quat-ef} at the end of Phase {\rm I} (i.e., at the termination of the first while-loop). Then, for all agents $i, j \in [n]$ and all unassigned intervals $U \in \U_{\overline{\mathcal{P}}}$, we have 
\begin{align*}
\widehat{v}_i(\overline{P}_i) \geq v_i(\overline{P}_j) - \frac{\delta}{n} \qquad \text{ and } \qquad 
\widehat{v}_i(\overline{P}_i) \geq v_i(U) - \frac{\delta}{n}.
\end{align*}
\end{lemma}
\begin{proof}
Fix an arbitrary agent $i \in [n]$ and consider any unassigned interval $U \in \U_{\overline{\mathcal{P}}}$. The execution condition of the first while-loop ensures that at termination it holds that $\widehat{v}_i(\overline{P}_i) \geq \widehat{v}_i(U) - \frac{\delta}{n} \geq v_i(U) - \frac{\delta}{n}$; the last inequality directly follows from the definition of $\widehat{v}_i$. This establishes the desired inequalities with respect to the unassigned intervals. 

Next, for any assigned interval $\overline{P_j}=[\overline{\ell}_j, \overline{r}_j]$, assume, towards a contradiction, that $v_i(\overline{P}_j) > \widehat{v}_i(\overline{P}_i) + {\delta}/{n}$.  Since the valuation $v_i$ is divisible (see Section \ref{section:notation}),\footnote{Here, we invoke divisibility of $v_i$ with factor $\alpha = \frac{\widehat{v}_i(\overline{P}_i) + {\delta}/{n}}{v_i(\overline{P}_j)} \in (0,1)$. Also, note that, in contrast to $v_i$, the function $\widehat{v}_i$ is not divisible.}  there exists a strict subset $X = [\overline{\ell}_j, x] \subsetneq \overline{P}_j$ with the property that $v_i(X) = \widehat{v}_i(\overline{P}_i) + {\delta}/{n}$. This, however, contradicts Lemma \ref{lemma:min-inclusion} (instantiated with $\mathcal{P}'=\overline{\mathcal{P}}$). The lemma stands proved.  
\end{proof}

Next, we show that the bounded envy guarantee obtained at the end of Phase {\rm I} (as stated in Lemma \ref{lem:envyfreephase1}) continues to hold in Phase {\rm II}.

\begin{lemma}\label{lem:envyfreephase2}
Let $\mathcal{P}=(P_1, \ldots, P_n)$ be the partial allocation maintained by Algorithm \ref{alg:quat-ef} at the end of Phase {\rm II}.  Then, for all agents $i, j \in [n]$ and all unassigned intervals $U \in \U_{\mathcal{P}}$, we have 
\begin{align*}
\widehat{v}_i(P_i) \geq v_i(P_j) - \frac{\delta}{n} \qquad \text{ and } \qquad 
\widehat{v}_i(P_i) \geq v_i(U) - \frac{\delta}{n} .
\end{align*}
\end{lemma}
\begin{proof}
Let $\overline{\mathcal{P}}=(\overline{P}_1, \ldots, \overline{P}_n)$ be the partial allocation maintained by Algorithm \ref{alg:quat-ef} at the end of Phase {\rm I}. Note that $\mathcal{U}_{\overline{P}}$ and $\mathcal{U}_\mathcal{P}$ denote the collection of unassigned intervals left at the end of Phase {\rm I} and Phase {\rm II}, respectively. We observe that, for all unassigned intervals $U \in \mathcal{U}_\mathcal{P}$, there exists an unassigned interval $\overline{U} \in \U_{\overline{\mathcal{P}}}$ such that $U \subseteq \overline{U}$. These containments follow from the fact that in each iteration of the second while-loop (i.e., in Phase {\rm II}) we either reassign the allocated intervals (which preserves the collection of the unassigned ones) or we enlarge a chosen assigned interval (Line \ref{step:crumb}); under such an enlargement, one of the unassigned intervals gets reduced and the others remain unchanged. 

Furthermore, Lemma \ref{lemma:hat-monotone} gives us $\widehat{v}_i(P_i) \geq \widehat{v}_i(\overline{P}_i) \geq v_i(\overline{U}) - \delta/n$, for any interval $\overline{U} \in \mathcal{U}_{\overline{P}}$; here, the last inequality follows from Lemma \ref{lem:envyfreephase1}. Using this bound and the above-mentioned containment of unassigned intervals we get $\widehat{v}_i(P_i) \geq v_i(U) - \delta/n $, for all $U \in \mathcal{U}_\mathcal{P}$. This establishes the desired inequalities with respect to the unassigned intervals. 

Next, for any assigned interval ${P_j}=[{\ell}_j, {r}_j]$, assume, towards a contradiction, that $v_i({P}_j) > \widehat{v}_i({P}_i) + {\delta}/{n}$.  Since the valuation $v_i$ is divisible,  there exists a strict subset $X = [{\ell}_j, x] \subsetneq {P}_j$ with the property that $v_i(X) = \widehat{v}_i({P}_i) + {\delta}/{n}$. This, however, contradicts Lemma \ref{lemma:min-inclusion} (instantiated with $\mathcal{P}'= {\mathcal{P}}$). The lemma stands proved.  
\end{proof}

The following lemma shows that, at the end of Phase {\rm I}, if an agent $i$ does not receive a bifurcating interval but some other agent $j$ does, then $j$'s interval cannot be bifurcating (for $i$) with an additional margin of $\frac{\delta}{n}$. Formally,\footnote{Note that, in contrast to Lemmas \ref{lem:envyfreephase1} and \ref{lem:envyfreephase2}, here we have an absolute bound on the value of the compared interval $\overline{P}_j$.} 
\begin{restatable}{lemma}{LemmaBifurcatingPone}
\label{lemma:bifurcating-phase-one}
Let $\overline{\mathcal{P}}=(\overline{P}_1, \ldots, \overline{P}_n)$ be the partial allocation maintained by Algorithm \ref{alg:quat-ef} at the end of Phase {\rm I}. If, for an agent $i\in [n]$, the assigned interval $\overline{P}_i$ is \emph{not} bifurcating (for $i$), but interval $\overline{P}_j = [\overline{\ell}_j, \overline{r}_j]$ is bifurcating (for $i$). Then, at least one of the following inequalities holds:  
\begin{align*}
v_i(\overline{P}_j) < \frac{1}{4} + \frac{\delta}{n} \qquad \text{ or } \quad v_i([\overline{r}_j, 1]) > \frac{1}{2} - \frac{\delta}{n}.
\end{align*}
\end{restatable}
\begin{proof}
Assume, towards a contradiction, that the bifurcating interval $\overline{P}_j = [\overline{\ell}_j, \overline{r}_j]$ has value $v_i(\overline{P}_j) \geq \frac{1}{4} + \frac{\delta}{n}$ and $v_i([\overline{r}_j, 1]) \leq \frac{1}{2} - \frac{\delta}{n}$. These properties in fact imply the existence of a strict subset $X  \subsetneq \overline{P}_j$ that is bifurcating for $i$: write $x \in [\overline{\ell}_j, \overline{r}_j] $ to denote the leftmost point that satisfies $v_i([\overline{\ell}_j, x]) =  v_i(\overline{P}_j) - \frac{\delta}{n}$ and set $X = [\overline{\ell}_j, x]$. Since valuation $v_i$ is divisible, such a point $x$ exists and we have $x < \overline{r}_j$. Furthermore,  the lower bound on the value of $\overline{P}_j$ gives us $v_i(X) \geq \frac{1}{4}$. In addition, note that $v_i([0,\overline{\ell}_j]) \leq 1/2$, since $\overline{P}_j = [\overline{\ell}_j, \overline{r}_j]$ is bifurcating.  Also, using the inequality $v_i([\overline{r}_j, 1]) \leq \frac{1}{2} - \frac{\delta}{n}$ and the additivity of the valuation $v_i$, we get that $v_i([x,1]) \leq \frac{1}{2}$. Indeed, these bounds ensure that $X$ is a strict subset of $\overline{P}_j$ and is bifurcating for agent $i$. 

The existence of $X$ contradicts the selection criterion in Lines \ref{step:candidates} and \ref{step:update}. In particular, consider the iteration in which $\overline{P}_j$ was assigned and write $P'_i$ to denote the interval assigned to agent $i$ during that iteration. We note that 
\begin{align*}
\widehat{v}_i(P'_i) & \leq \widehat{v}_i(\overline{P}_i) \tag{via Lemma \ref{lemma:hat-monotone}} \\
& < \frac{1}{2} \tag{$\overline{P}_i$ is non-bifurcating \& Claim \ref{claim:non-bifurcating}}
\end{align*} 
On the other hand, $\widehat{v}_i(X) = 1$, for the interval $X$ identified above. Hence, $j$ would not be the selected agent in Line \ref{step:update}. This contradiction establishes the lemma. 
\end{proof}

We next prove that a guarantee, analogous to Lemma \ref{lemma:bifurcating-phase-one}, holds for Phase {\rm II} as well.\footnote{As in Lemma \ref{lemma:bifurcating-phase-one}, here we have an absolute bound on the value of the compared interval ${P}_j$.}

\begin{restatable}{lemma}{LemmaBifurcatingPtwo}
\label{lemma:bifurcating-phase-two}
Let $\mathcal{P}=(P_1, \ldots, P_n)$ be the partial allocation maintained by Algorithm \ref{alg:quat-ef} at the end of Phase {\rm II}. If, for an agent $i\in [n]$, the assigned interval $P_i$ is \emph{not} bifurcating (for $i$), but interval ${P}_j = [{\ell}_j, {r}_j]$ is bifurcating (for $i$). Then, at least one of the following inequalities holds: 
\begin{align*}
v_i(P_j) < \frac{1}{4} + \frac{\delta}{n}  \qquad \text{ or } \quad v_i([r_j, 1]) > \frac{1}{2} - \frac{\delta}{n}.
\end{align*}
\end{restatable}
\begin{proof}
Write $\overline{\mathcal{P}}=(\overline{P}_1, \ldots, \overline{P}_n)$ to denote the partial allocation at the end of Phase {\rm I}. Note that, throughout Phase {\rm II}, the algorithm either reassigns the intervals (Line \ref{step:envycycle}) or appends (unassigned) pieces to them (Line \ref{step:crumb}). Hence, for the interval $P_j$, assigned to agent $j$ at the end of Phase {\rm II}, there exists $\overline{P}_k$, for some $k \in [n]$, such that $P_j \supseteq \overline{P}_k$. 

We assume, towards a contradiction, that the bifurcating interval ${P}_j = [{\ell}_j, {r}_j]$ has value $v_i({P}_j) \geq \frac{1}{4} + \frac{\delta}{n}$ and $v_i([{r}_j, 1]) \leq \frac{1}{2} - \frac{\delta}{n}$. It cannot be the case that $P_j = \overline{P}_k$ (for an interval $\overline{P}_k$ assigned at the end of Phase {\rm I}), since this would contradict Lemma \ref{lemma:bifurcating-phase-one}. Hence, in the remainder of the proof we address the complementary case wherein $P_j$ was obtained by appending to an interval, say $P'_s$, in an iteration of the second while-loop (specifically, Line \ref{step:crumb}). Also, write $P'_i$ to denote the interval assigned to agent $i$ during that iteration. Lemma \ref{lemma:hat-monotone} and the fact that $P_i$ is non-bifurcating for $i$ (Claim \ref{claim:non-bifurcating}) give us $\widehat{v}_i (P'_i) \leq \widehat{v_i}(P_i) < 1/2$. Using this inequality and the fact that $s$ was a source agent during the iteration under consideration, we get 
\begin{align}
\widehat{v}_i (P'_s) \leq \widehat{v}_i(P'_i) < 1/2 \label{ineq:shalf}
\end{align}
However, the assumptions on the values of $P_j = [{\ell}_j, {r}_j]$ and $[r_j,1]$ contradict inequality (\ref{ineq:shalf}): Let $z$ denote the right endpoint of $P'_s$, i.e., $P'_s = [\ell_j, z]$. Since a piece of bounded value is appended in Line \ref{step:crumb}, we have $v_i(P'_s) \geq v_i(P_j) - \frac{\delta}{n} \geq \frac{1}{4}$. Furthermore, using the inequality $v_i([{r}_j, 1]) \leq \frac{1}{2} - \frac{\delta}{n}$, we obtain $v_i([z, 1]) \leq \frac{1}{2}$. In addition, the fact that $P_j= [{\ell}_j, {r}_j]$ is a bifurcating interval gives us $v_i([0,\ell_j)) \leq \frac{1}{2}$, i.e., the value to the left of $P'_s = [\ell_j, z]$ is at most $1/2$. These observations imply that $P'_s$ is a bifurcating interval for $i$; in particular, $\widehat{v}_i(P'_s) = 1$. This bound contradicts inequality (\ref{ineq:shalf}) and completes the proof.  
\end{proof}

Using Lemma \ref{lem:envyfreephase2}, we next obtain a relevant envy bound for the allocation $\mathcal{I}$ returned by the algorithm. 
\begin{lemma}\label{lem:half-mult}
The allocation $\mathcal{I} = (I_1, \ldots, I_n)$ computed by Algorithm \ref{alg:quat-ef} satisfies $v_i(I_i) \geq \frac{1}{2}v_i(I_j)  - \frac{\delta}{n}$, for all agents $i, j \in [n]$. 
\end{lemma}
\begin{proof}
Write $\mathcal{P}=(P_1, \ldots, P_n)$ to denote the partial allocation of the algorithm at the end of Phase {\rm II}. The execution condition of the second while-loop ensures that $|\U_\mathcal{P}| \leq n$. Also, at the end of the algorithm, for each unassigned interval $U \in \U_\mathcal{P}$, we select a distinct and adjacent interval $P_j$ and associate $U$ with $P_j$. In particular, let $U_j$ be the unassigned interval associated with $P_j$. If $P_j$ is not associated with any unassigned interval, then set $U_j = \emptyset$. Indeed, $I_j = P_j \cup U_j$ is a connected piece of the cake, i.e., an interval. 

For any agent $i \in [n]$, if interval $I_i$ is bifurcating, then $v_i(I_i) \geq \frac{1}{4}$. Furthermore, any other assigned interval $I_j$ is either completely to the left of $I_i$ or to the right of $I_i$. In either case, by the definition of bifurcating intervals, we have $v_i(I_j) \leq \frac{1}{2}$. Hence, for agents $i$ that receive a bifurcating interval $I_i$, we have the stated inequality, $v_i(I_i) \geq \frac{1}{2} v_i(I_j)$. 

It remains to show that the lemma holds for agents $i$ for whom $I_i$ is not bifurcating. For such agents, the interval $P_i \subseteq I_i$ is also non-bifurcating and, hence, $v_i(P_i) = \widehat{v}_i(P_i)$. Therefore, 
\begin{align*}
v_i(I_i) \geq v_i(P_i) \geq v_i(P_j) - \frac{\delta}{n} \quad \text{ and } \quad v_i(I_i) \geq v_i(P_i)  \geq v_i(U_j) - \frac{\delta}{n} && \text{(via Lemma \ref{lem:envyfreephase2})}
\end{align*}
Summing we get 
\begin{align*}
 2  v_i(I_i) \geq v_i(P_j) + v_i(U_j) -   \frac{2\delta}{n} = v_i(I_j)  -   \frac{2\delta}{n} \tag{since $v_i$ is additive}
\end{align*}
Hence, we obtain the stated inequality, $v_i(I_i) \geq \frac{1}{2}v_i(I_j)  - \frac{\delta}{n}$. This completes the proof.
\end{proof}

We now establish the main result of this section. 

\begin{theorem}
\label{theorem:add-ef}
Given any cake division instance---with Robertson-Webb query access to the valuations of the $n$ agents---and parameter $\delta \in (0,1)$, Algorithm \ref{alg:quat-ef} computes a connected cake division (i.e., an allocation) $\mathcal{I}=(I_1, \ldots, I_n)$ that is $\left(\frac{1}{4} + \frac{2\delta}{n}\right)$-$\EF$. The algorithm executes in time that is polynomial in $n$ and $\frac{1}{\delta}$. 
\end{theorem}
\begin{proof}
Fix any agent $i \in [n]$. We establish the theorem by considering three complementary and exhaustive cases, based on the interval $I_i$ (assigned to agent $i$): \\
\noindent 
Case {\sf 1}: Interval $I_i$ is bifurcating, \\
\noindent
Case {\sf 2}: Value $v_i(I_i) < \frac{1}{4}$, \\
\noindent
Case {\sf 3}: Interval $I_i$ is not bifurcating and $v_i(I_i) \geq \frac{1}{4}$. 

In Case {\sf 1}, since interval $I_i$ is bifurcating for agent $i$, we have $v_i(I_i) \geq \frac{1}{4}$ and, for any other interval $I_j$ (either to the left of $I_i$ or to its right), we have $v_i(I_j) \leq \frac{1}{2}$. Therefore, in this case, the stated approximation bound on envy holds, $v_i(I_i) \geq v_i(I_j) - \frac{1}{4}$. 

In Case {\sf 2}, value $v_i(I_i) < \frac{1}{4}$. Note that, Lemma \ref{lem:half-mult} gives us $v_i(I_i) \geq \frac{1}{2}  v_i(I_j) - \frac{\delta}{n}$, for any other agent $j \in [n]$. Multiplying both sides of this inequality by $2$ and simplifying we obtain  
\begin{align*}
v_i(I_i) &\geq v_i(I_j) - v_i(I_i) - \frac{2\delta}{n} \geq v_i(I_j) - \frac{1}{4} - \frac{2\delta}{n} \tag{since $v_i(I_i) < \frac{1}{4}$}
\end{align*}
Therefore, in Case {\sf 2} as well, for agent $i$ the envy is additively at most $\left(\frac{1}{4} + \frac{2\delta}{n}\right)$. 
\begin{figure}[h]
	\begin{center}
		\includegraphics[scale=.6]{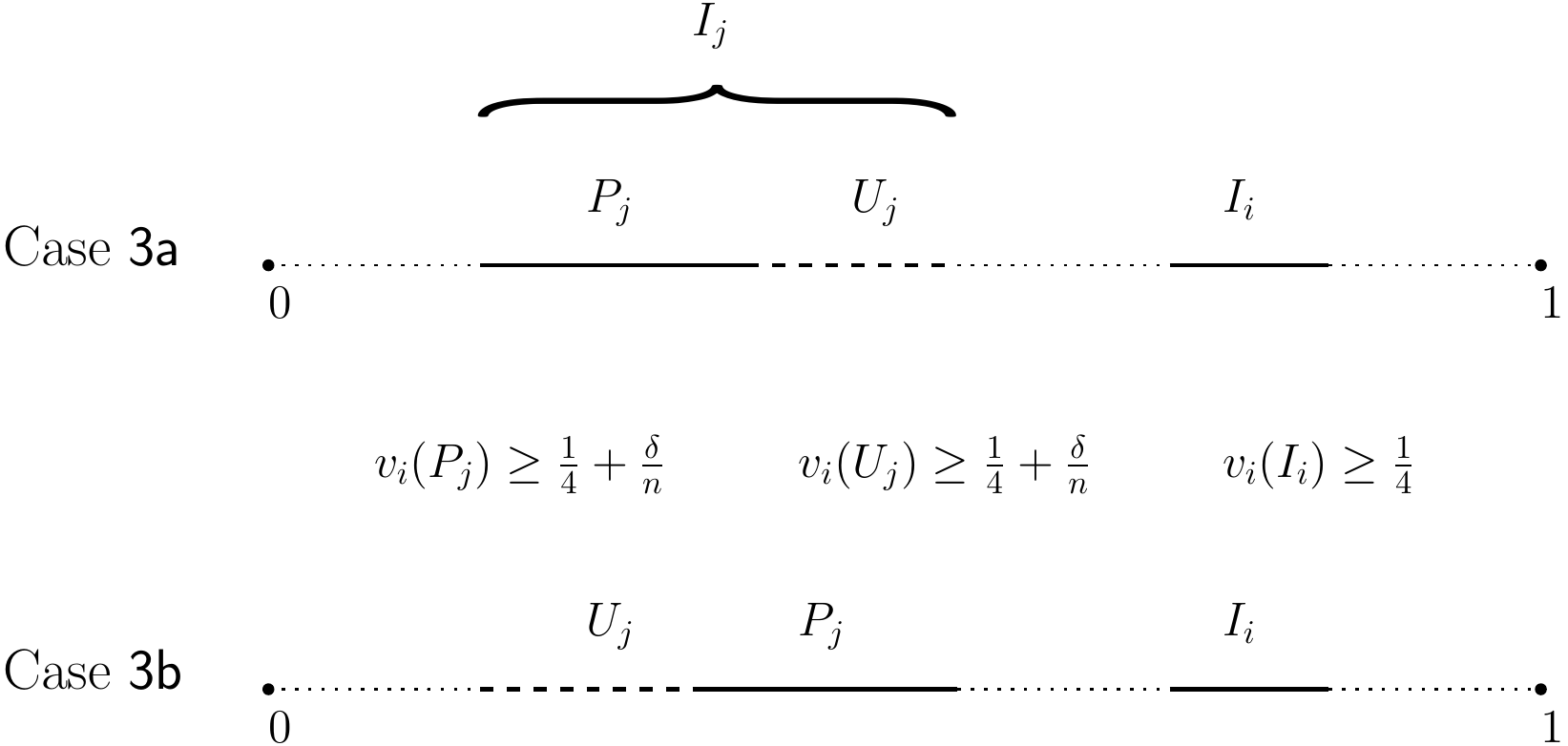}
	\end{center}
	\vspace*{-17pt}
	\caption{Placement of intervals $P_j$ and $U_j$ in Case {\sf 3}.}
    \label{fig:Case3}
\end{figure}

Finally, in Case {\sf 3}, interval $I_i$ is not bifurcating and $v_i(I_i) \geq \frac{1}{4}$. Write $\mathcal{P}=(P_1, \ldots, P_n)$ to denote the partial allocation at the end of Phase {\rm II} and recall that $I_k = P_k \cup U_k$, for each agent $k \in [n]$ and the associated unassigned interval $U_k \in \mathcal{U}_\mathcal{P}$. For analyzing this case, assume, towards a contradiction, that there exists an interval $I_j$ that violates the stated approximate envy-freeness bound, i.e., 
\begin{align}
v_i(I_j) > v_i(I_i) + \frac{1}{4} +  \frac{2\delta}{n}.\label{ineq:for-contra}
\end{align}
Since, in the current case, $v_i(I_i) \geq \frac{1}{4}$, inequality (\ref{ineq:for-contra}) reduces to $v_i(I_j) > \frac{1}{2} + \frac{2\delta}{n}$. We will further show that for interval $I_j$ the composing sub-intervals $P_j$ and $U_j$ are each of value (under $v_i$) at least $\frac{1}{4} + \frac{\delta}{n}$. Towards this, note that, in the current case, since $I_i$ is not bifurcating for $i$, neither is $P_i \subseteq I_i$.  Furthermore, 
\begin{align}
v_i(I_i) &\geq v_i(P_i) \tag{$v_i$ is monotonic} \\
&= \widehat{v}_i(P_i) \tag{since $P_i$ is not bifurcating for $i$} \\
& \geq  v_i(P_j) - \frac{\delta}{n}  \label{ineq:int}
\end{align}
The last inequality follows from Lemma \ref{lem:envyfreephase2}. A similar application of the lemma also gives us 
\begin{align}
v_i(I_i) & \geq v_i(U_j) - \frac{\delta}{n}  \label{ineq:inta}
\end{align}
Inequalities (\ref{ineq:for-contra}), (\ref{ineq:int}), and (\ref{ineq:inta})  imply that the values of both $P_j$ and $U_j$ are at least $\frac{1}{4} + \frac{\delta}{n}$. Otherwise, say $v_i(P_j) < \frac{1}{4} + \frac{\delta}{n}$. Then, 
\begin{align*}
v_i(I_j) = v_i(P_j) + v_i(U_j) < \frac{1}{4} + \frac{\delta}{n} + v_i(U_j) \leq v_i(I_i) + \frac{1}{4} + \frac{2 \delta}{n} \tag{via inequality (\ref{ineq:inta})}
\end{align*} 
Since the last inequality contradicts assumption (\ref{ineq:for-contra}), we have $v_i(P_j) \geq \frac{1}{4} + \frac{\delta}{n}$. Similarly, $v_i(U_j) \geq \frac{1}{4} + \frac{\delta}{n}$. 

As mentioned previously, $v_i(I_j) > \frac{1}{2} + \frac{2 \delta}{n}$. Hence, the values to the left and to the right of $I_j=[\ell_j, r_j]$ are upper bounded as follows:\footnote{Recall that the value of the entire cake is normalized, $v_i([0,1])=1$.} $v_i([0,\ell_j]) \leq \frac{1}{2} - \frac{2\delta}{n}$ and $v_i([r_j,1]) \leq \frac{1}{2} - \frac{2\delta}{n}$. 

For the subsequent analysis, we also assume that interval $I_j$ is on the left of $I_i$ (see Figure \ref{fig:Case3}); the proof for the other configuration (of $I_j$ being to the right of $I_i$) follows analogously. Now, there are two sub-cases to consider: \\
\noindent
Case {\sf 3a}: Interval $P_j$ is to the left of $U_j$ (i.e., $U_j$ lies between $P_j$ and $I_i$). \\
\noindent 
Case {\sf 3b}: Interval $P_j$ is to the right of $U_j$ (i.e., $P_j$ lies between $U_j$ and $I_i$). \\

In Case {\sf 3a}, we note that the interval $U_j \in \mathcal{U}_\mathcal{P}$ is bifurcating for agent $i$: As observed above, $v_i(U_j) \geq \frac{1}{4} + \frac{\delta}{n}$ and the value (in the cake) to the right of $U_j$ is equal to $v_i([r_j, 1]) \leq \frac{1}{2} - \frac{2\delta}{n}$. In addition, the value to the left of $U_j$ is at most $1 - (v_i(U_j) + v_i(I_i)) \leq 1 - \frac{1}{4}  - \frac{1}{4} - \frac{\delta}{n} = \frac{1}{2} - \frac{\delta}{n}$; interval $I_i$ is to the right to $I_j$ and, hence, to the right of $U_j$. Hence, $U_j \in \mathcal{U}_\mathcal{P}$ is bifurcating for agent $i$. 

Also, the design of Phase {\rm II} ensures that, for the interval $U_j$, there exists an unassigned interval $\overline{U} \in \mathcal{U}_{\overline{\mathcal{P}}}$ such that $\overline{U} \supseteq U_j$; here $\overline{\mathcal{P}}=(\overline{P}_1, \ldots, \overline{P}_n)$ denotes the partial allocation at the end of Phase {\rm I}. Since $U_j$ is bifurcating for $i$, so is $\overline{U}$. By contrast, in the current case (Case {\sf 3}), the interval $P_i$ is not bifurcating for $i$ and, hence, neither is $\overline{P}_i$ (Lemma \ref{lemma:hat-monotone}). That is, $\widehat{v}_i (\overline{P}_i) < \frac{1}{2} < 1 = \widehat{v}_i (\overline{U})$. The bound, however, contradicts the termination of the first while-loop. Therefore, by way of contradiction, we get that assumption (\ref{ineq:for-contra}) cannot hold in Case {\sf 3a}. This completes the analysis of this sub-case. \\

In Case {\sf 3b}, we note that the interval $P_j$ is bifurcating for agent $i$, with a margin of $\frac{\delta}{n}$: As observed above, $v_i(P_j) \geq \frac{1}{4} + \frac{\delta}{n}$ and the value to the right of $P_j$ is equal to $v_i([r_j, 1]) \leq \frac{1}{2} - \frac{2\delta}{n}$. In addition, the value to the left of $P_j$ is at most $1 - (v_i(P_j) + v_i(I_i)) \leq 1 - \frac{1}{4}  - \frac{1}{4} - \frac{\delta}{n} = \frac{1}{2} - \frac{\delta}{n}$. The existence of such a bifurcating interval $P_j$ contradicts Lemma \ref{lemma:bifurcating-phase-two}. Hence, even in Case {\sf 3b}, we must have $v_i(I_j) \leq v_i(I_i) + \frac{1}{4} + \frac{2\delta}{n}$, i.e., the stated bound on envy holds.  

This completes the analysis for all the cases, and the theorem stands proved. 
\end{proof}

Complementing the additive envy-freeness guarantee obtained in Theorem \ref{theorem:add-ef}, the next result establishes that, in the computed allocation $\mathcal{I}$, the envy is within a factor of $(2+c)$, where parameter $c \in (0,1)$ is polynomially small (in $n$). 
\begin{theorem}
\label{theorem:mult-ef}
Given any cake division instance---with Robertson-Webb query access to the valuations of the $n$ agents---and parameter $c \in (0,1)$, we can compute (in time polynomial in $n$ and $1/c$) a connected cake division (i.e., an allocation) $\mathcal{I}=(I_1, \ldots, I_n)$ that is $\frac{1}{2+c}$-mult-$\EF$. 
\end{theorem}
\begin{proof}
The theorem directly follows from Lemma \ref{lem:half-mult}. In particular, we execute Algorithm \ref{alg:quat-ef} with parameter $\delta = \frac{c}{8}$, for a sufficiently small $c \in (0,1)$,\footnote{With this parameter choice, the algorithm executes in time that is polynomial in $n$ and $1/c$. } and note that, for the computed allocation $\mathcal{I}$, Lemma \ref{lem:half-mult} gives us $v_i(I_i) \geq \frac{1}{2}v_i(I_j)  - \frac{\delta}{n} = \frac{1}{2}v_i(I_j)  - \frac{c}{8n}$, for agents $i, j \in [n]$. Summing over $j$, we obtain 
\begin{align}
n \ v_i(I_i) \geq \frac{1}{2} \sum_{j=1}^n v_i(I_j)   -  \frac{c}{8} = \frac{1}{2} - \frac{c}{8} \label{ineq:near-prop}
\end{align}
The last equality follows from the fact that $I_1, \ldots, I_n$ constitute a complete partition of the cake, with value $v_i([0,1]) = 1$. Since constant $c \leq 1$, inequality (\ref{ineq:near-prop}) reduces to $v_i(I_i) \geq \frac{1}{4n}$, for all agents $i \in [n]$. Therefore, the bound obtained in Lemma \ref{lem:half-mult} can be expressed as 
\begin{align*}
v_i(I_i)  \geq \frac{1}{2}v_i(I_j)  - \frac{c}{8n} \geq \frac{1}{2}v_i(I_j)  - \frac{c}{2} v_i(I_i).
\end{align*}
Simplifying we obtain $\left( 2 + c \right) v_i(I_i) \geq v_i(I_j)$, for all agents $i, j \in [n]$. Therefore, the computed allocation is $\frac{1}{2+c}$-mult-$\EF$
\end{proof}
\section{An $\varepsilon$-$\EF$ Algorithm under Bounded Heterogeneity}\label{sec:boundedhet}
\label{section:bounded-het}
This section addresses cake division instances in which, for a parameter $\varepsilon \in (0,1)$ and across the $n$ agents, the number of distinct valuations is at most $(\varepsilon n-1)$. Our algorithm (Algorithm \ref{alg:boundedhetero}) for finding $\varepsilon$-$\EF$ allocations in such instances in detailed next.  

\begin{algorithm}[h]
\caption{$\varepsilon$-$\EF$ under bounded heterogeneity} \label{alg:boundedhetero}
\textbf{Input:} A cake division instance with oracle access to the valuations $\{v_i\}_{i=1}^n$ of the $n$ agents along with parameter $\varepsilon \in (0,1)$.\\ 
\textbf{Output:} A complete allocation $(I_1,\ldots, I_n)$. 
\begin{algorithmic}[1]
\STATE Set $T$ to be the smallest integer such that $T \varepsilon \geq 1$, i.e., $T \coloneqq \left\lceil \frac{1}{\varepsilon}\right\rceil$.
\STATE \label{step:cut-pts} For each agent $i \in [n]$, let $0 = x^i_0  < x^i_1 < x^i_2 < \ldots < x^i_{T-1} < x^i_T = 1$ be the collection of $(T+1)$ cut points that satisfy $v_i([x^i_{t-1}, x^i_{t}]) = \varepsilon$, for all $1 \leq t \leq T-1$, and $v_i([x^i_{T-1}, x^i_{T}]) \leq \varepsilon$. 
\STATE \label{step:Z-pop} Let $Z$ be the union of these cut points $Z \coloneqq \bigcup_{i \in [n]} \left\{x^i_0, x^i_1, \ldots, x^i_{T-1}, x^i_T \right\}$ 
\\ \COMMENT{$Z$ is not a multiset, i.e., multiple instances of same cut point are not repeated in $Z$.}
\STATE \label{step:populate-F} Index the points in $Z = \left\{z_0, z_1, z_2, \ldots, z_r \right\}$ such that $0= z_0  < z_1 < z_2 < \ldots < z_r =1$ and define the collection of intervals $\mathcal{F} \coloneqq \Big\{ [z_t, z_{t+1} ] \Big\}_{t=0}^{r-1}$ 
\FOR{agents $i = 1$ to $n$}
\STATE If $\mathcal{F} = \emptyset$, then set interval $I_i = \emptyset$. Otherwise, if $\mathcal{F} \neq \emptyset$, then set $I_i = \argmax_{F \in \mathcal{F}} \ v_i(F)$ and update $\mathcal{F} \gets \mathcal{F} \setminus \{ I_i \}$. 
\ENDFOR
\RETURN allocation $\mathcal{I} = (I_1, \ldots, I_n)$.				
\end{algorithmic}
\end{algorithm}
We first show that the collection of intervals computed by Algorithm \ref{alg:boundedhetero} cover the entire cake. We will then use this lemma to establish the approximate envy-freeness guarantee in Theorem \ref{theorem:bounded-het}.  

\begin{lemma}\label{lem:boundedheterofull}
Given any cake division instance in which, across the $n$ agents, the number of {distinct} valuations is at most $\left(\varepsilon n -1\right)$, Algorithm \ref{alg:boundedhetero}'s output $\mathcal{I} = (I_1, \ldots, I_n)$ is a complete allocation, i.e., $I_i$s are pairwise disjoint and $\cup_{i \in [n]} I_i = [0,1]$.  
\end{lemma}
\begin{proof}
By construction, the set of intervals $\mathcal{F}$ populated in Line \ref{step:populate-F} of Algorithm \ref{alg:boundedhetero} are pairwise disjoint and cover the entire cake. We will show that the number of intervals in $\mathcal{F}$ is at most $n$, i.e., $|\mathcal{F}| \leq n$. Since the assigned intervals, $I_i$s, are selected from the set $\mathcal{F}$ (see the for-loop in the algorithm), the cardinality bound implies that no interval in $\mathcal{F}$ remains unassigned. Hence, $\cup_{i \in [n]} I_i = [0,1]$. Also, given that the intervals in $\mathcal{F}$ are pairwise disjoint, so are the $I_i$s. Therefore,  $\mathcal{I} = (I_1, \ldots, I_n)$ is a complete allocation. 

We complete the proof by establishing that $|\mathcal{F}| \leq n$. Towards this it suffices to show that $|Z| \leq n +1$; see Line \ref{step:populate-F} and note that $|\mathcal{F}| = |Z|-1$. In Line \ref{step:cut-pts}, for each agent $i \in [n]$, we consider $T+1$ cut points $0 = x^i_0  < x^i_1 < x^i_2 < \ldots < x^i_{T-1} < x^i_T = 1$. The end points of the cake, $0$ and $1$, are considered for every agent. Moreover, for any two agents, $i, j \in [n]$, with identical valuations, $v_i = v_j$, even the remaining $(T-1)$ points are the same: $x^i_t = x^j_t$ for all $1 \leq t \leq T-1$. Since the number of distinct valuations is at most $(\varepsilon n -1)$, there are at most $(\varepsilon n -1) (T-1)$ cut points in $Z$ that are strictly between $0$ and $1$. Including the endpoints of the cake in the count, we get $|Z| \leq (\varepsilon n -1) \left( T - 1\right)+2 \leq (\varepsilon n - 1) \frac{1}{\varepsilon} + 2$. The last inequality follows from the definition of $T$; in particular, $T -1 < \frac{1}{\varepsilon}$. Simplifying we obtain $|Z| \leq n - \frac{1}{\varepsilon} + 2 \leq n+1$; recall that $\varepsilon \leq 1$. Therefore, $|\mathcal{F}| \leq n$ and the lemma stands proved.  
\end{proof}

The following theorem establishes that Algorithm \ref{alg:boundedhetero} finds an allocation $\mathcal{I} = (I_1, \ldots, I_n)$ that satisfies $v_i(I_i) \geq v_i(I_j) - \varepsilon$ for all agents $i, j \in [n]$.

\begin{restatable}{theorem}{BoundedHetero}
\label{theorem:bounded-het}
Given any cake division instance in which, across the $n$ agents, the number of {distinct} valuations is at most $\left(\varepsilon n -1\right)$, Algorithm $\ref{alg:boundedhetero}$ (with Robertson-Webb query access to the valuations) computes an $\varepsilon$-$\EF$ allocation in polynomial time. 
\end{restatable}

\begin{proof}
The runtime analysis of the algorithm is direct. Also, via Lemma \ref{lem:boundedheterofull}, we have that the the returned tuple $\mathcal{I} = (I_1, \ldots, I_n)$ is indeed a complete allocation. 

For proving that the algorithm achieves an $\varepsilon$-$\EF$ guarantee, consider any agent $i \in [n]$ and interval $I_j = [z_t, z_{t+1}]$, where $z_t$ and $z_{t+1}$ are successive points in the set $Z$; see Line \ref{step:populate-F}. If $I_j = \emptyset$, then in fact $i$ does not envy $j$. We will show that $v_i(I_j) \leq \varepsilon$ and, hence, obtain the desired bound: $v_i(I_i) \geq v_i(I_j) - \varepsilon$.  

For agent $i$, write $x^i_s$ to be the largest (rightmost) cut point considered in Line \ref{step:cut-pts} that satisfies $x^i_s \leq z_t$. In particular, $x^i_{s+1} > z_t$. Note that the set $Z$ (see Line \ref{step:Z-pop}) contains all the points $x^i_0, x^i_1, \ldots, x^i_{T-1}, x^i_T$; in particular, $x^i_{s+1} \in Z$. In addition, $z_t$ and $z_{t+1}$ are two successive points in $Z$. Hence, we have $z_{t+1} \leq x^i_{s+1}$ and the interval $I_j = [z_t, z_{t+1}] \subseteq [x^i_s, x^i_{s+1}]$. By construction, $v_i([x^i_s, x^i_{s+1}]) \leq \varepsilon$ and, hence, $v_i(I_j) \leq \varepsilon$.  This bound on the valuation of interval $I_j$ implies that the computed allocation is $\varepsilon$-$\EF$. The theorem stands proved. 
\end{proof}

Note that the algorithm might assign some agents $i \in [n]$ an interval of value of zero; in particular, $I_i = \emptyset$. Imposing the requirement that each agent $i \in [n]$ receives an interval of nonzero value (to $i$) renders the problem as hard as finding an $\varepsilon$-$\EF$ allocation in general cake division instances (without bounded heterogeneity). To see this, consider any cake division instance (which might not satisfy the bounded heterogeneity condition).  Append to the cake another unit length interval and include $\left \lceil \frac{n+2}{\varepsilon}\right\rceil$ dummy agents that have identical valuation confined to the appended interval. This new instance satisfies bounded heterogeneity. Now, if each agent in the constructed instance receives an interval of nonzero value, then the appended interval must have been divided among the dummy agents and the underlying cake $[0,1]$ among the original agents. This way we obtain an $\varepsilon$-$\EF$ allocation for the original instance. Furthermore, note that an $\alpha$-mult-$\EF$ guarantee, for any $\alpha >0$, implies that each agent receives an interval of nonzero value. Therefore, achieving multiplicative approximation bounds for envy under bounded heterogeneity is as hard as the general case. 

\begin{remark}
The discretization method used in Algorithm \ref{alg:boundedhetero} has been utilized in prior works as well; see \cite{branzei2017query} and \cite{lipton2004approximately}. However, the relevant insight obtained here is the difference between additive and multiplicative approximations: While one can efficiently achieve an $\varepsilon$-additive approximation under bounded heterogeneity, establishing any multiplicative guarantee is as hard as solving the problem in complete generality.
\end{remark}
\section{Conclusion and Future Work}
Algorithmically, connected envy-free cake division is a challenging and equally intriguing problem at the core of fair division. The proof of existence of envy-free cake divisions (with connected pieces) does not lend itself to efficient (approximation) algorithms and, at the same time, negative results---that rule out, say, a polynomial-time approximation scheme (PTAS)---are not known either. In this landscape, the current work improves upon the previously best-known approximation guarantee for connected envy-free cake division. We develop a computationally efficient algorithm that finds a connected cake division that is simultaneously $(1/4 + o(1))$-$\EF$ and $(1/2 - o(1))$-mult-$\EF$. 
We also show that specifically for instances with bounded heterogeneity, an $\varepsilon$-$\EF$  division can be computed in time polynomial in $n$ and $1/\varepsilon$.   

In addition to the patent problem of efficiently finding $\varepsilon$-$\EF$ connected cake divisions, developing $\varepsilon$-$\EF$ algorithms for special valuation classes (such as single-block and single-peaked valuations) is a relevant direction of future work. Inapproximability results---similar to the ones recently obtained for $\varepsilon$-consensus halving \cite{filos2020consensus}---are also interesting.

\bibliographystyle{alpha}
\bibliography{references}
\newpage
\appendix
\section{Missing Proofs from Section \ref{sec:mainsec}} \label{app:mainsec}
Here, we restate and prove Lemma \ref{lemma:cycle-elimination}.
\lemCycleElimination*
\begin{proof}
If, for given partial allocation $\mathcal{P}=(P_1,\ldots,P_n)$, the envy-graph $G_\mathcal{P}$ is already acyclic, then we directly obtain the lemma by setting $\mathcal{Q} = \mathcal{P}$. Hence, in the remainder of the proof we consider the case wherein $G_\mathcal{P}$ is cyclic. 

Write $C = i_1 \rightarrow i_2 \rightarrow \ldots \rightarrow i_k \rightarrow i_1$ to denote a cycle in $G_\mathcal{P}$. To obtain a new partial allocation $\mathcal{P'} = (P'_1, \ldots, P'_n)$, we  reassign the intervals as follows: for all agents $j$ not in the cycle (i.e., $j \notin \{i_1, i_2, \ldots, i_k\}$), set $P'_j = P_j$. Furthermore, for all the agents $i_t$ in the cycle $C$, with $1 \leq t <k $, we set $P'_{i_t} = P_{i_{t+1}}$ and $P'_{i_k} = P_{i_1}$. That is, each agent in the cycle receives the interval assigned to its successor in the cycle. This reassignment ensures that, for all agents $i \in [n]$, we have $\widehat{v}_i (P'_i) \geq \widehat{v}_i(P_i)$; recall that a directed edge $(i,j)$ is included in the graph $G_\mathcal{P}$ iff $\widehat{v}_i(P_i) < \widehat{v}_i(P_j)$. 

 We will now show that the number of edges in the envy-graph $G_{\mathcal{P'}}$ is strictly smaller than the number of edges in $G_\mathcal{P}$. Hence, repeated elimination of cycles leads to an allocation $\mathcal{Q}$ that satisfies the lemma. Note that the collection of intervals assigned in the allocation $\mathcal{P}'$ is the same as the collection of intervals in $\mathcal{P}$. Also, the out-degree of any vertex $i$ in $G_\mathcal{P}$ (or in $G_\mathcal{P'}$) is equal to the number of bundles $P_j$s (or $P'_j$s) have value (under $\widehat{v}_i$) strictly greater than $i$'s value (again, under $\widehat{v}_i$) for her own bundle. These observations imply that for all agents not in the cycle $C$, the out-degree is the same in $G_\mathcal{P}$ and $G_\mathcal{P'}$. Moreover, for all agents $i_t$ in the cycle $C$, we have $\widehat{v}_{i_t}(P'_{i_t}) > \widehat{v}_{i_t}(P_{i_t})$. Hence, the out-degree of any such agent $i_t$ in $G_\mathcal{P}$ is strictly smaller than its out-degree in $G_\mathcal{P}$. Therefore, the number of edges in $G_{\mathcal{P'}}$ is strictly smaller than the ones in $G_\mathcal{P}$. This strict reduction in the number of edges implies that after a polynomial number of cycle eliminations we obtain an allocation $\mathcal{Q}$ for which $G_\mathcal{Q}$ is acyclic and we have $\widehat{v}_i(Q_i) \geq \widehat{v}_i(P_i)$, for all agents $i \in [n]$. 
 The lemma stands proved. 
 \end{proof}

Claim \ref{clm:rwhatvs} is proved next. 

\ClaimRWHatV*
\begin{proof}
We first address the evaluation query for $\widehat{v}_i$. Given interval $[x,y] \subseteq [0,1]$, we use $\Eval_i$ to obtain the following three values: $v_i([x,y])$, $v_i([0,x])$, and $v_i([y,1])$. These three values tell us whether $[x,y]$ is a bifurcating interval for agent $i$; see Definition \ref{defn:bifurcating}. If $[x,y]$ is a bifurcating interval, then we have $\widehat{v}_i([x,y]) = 1$. Otherwise, $\widehat{v}_i([x,y]) = v_i([x,y])$.  

Now, we consider the cut query for $\widehat{v}_i$. Given a point $x \in [0, 1]$, and a value $\nu \in [0, 1]$, we identify two candidate points $y_1$ and $y_2$ and set $y \coloneqq \min\{y_1, y_2\}$ as the leftmost point that satisfies $\widehat{v}_i([x,y]) \geq \nu$. The first candidate point is defined as $y_1 \coloneqq \Cut_i(x,\nu)$, i.e., $y_1$ is the leftmost point that satisfies $v_i([x,y_1]) = \nu$. The definition of $\widehat{v}_i$ (see equation (\ref{defn:hat-v})) implies that $\widehat{v}_i([x,y_1]) \geq v_i([x,y_1)] = \nu$. Still, there could be a point $y_2$ to the left of $y_1$ such that the interval $[x,y_2]$ is bifurcating for $i$ and, hence, $\widehat{v}_i([x,y_2]) \geq \nu$. Therefore, the second candidate $y_2$ is computed by finding the smallest bifurcating interval, if one exists, starting at $x$. Towards this, we first use the query $\Eval_i(0,x)$ to ensure that $v_i([0,x]) \leq \frac{1}{2}$. If the interval $[0,x]$ is of value 
more than $1/2$, then we set $y_2 = 1$. In case $v_i([0,x]) \leq \frac{1}{2}$, we set $y_2 \coloneqq \max\left\{ \Cut_i(x, 0.25), \Cut_i(0, 0.5) \right\}$. Finally, we return the minimum of $y_1$ and $y_2$ as the answer $y$ to the cut query for $\widehat{v}_i$.  

Overall, we get that both the cut and the evaluation queries for $\widehat{v}_i$ can be answered in polynomial time. This completes the proof. 
\end{proof}

\end{document}